\documentclass[12pt]{article}
\usepackage[english]{babel}
\RequirePackage[colorlinks,citecolor=blue,urlcolor=blue,linkcolor=blue]{hyperref}
\hypersetup{
colorlinks = true,
citecolor=blue,
urlcolor=blue,
linkcolor=blue,
pdfauthor = {Edward Furman, Alexey Kuznetsov, Jianxi Su and Ricardas Zitikis},
pdfkeywords = {Gaussian copula, tail dependence, maximal dependence path},
pdftitle = {Tail dependence of the Gaussian copula revisited},
pdfpagemode = UseNone
}
\usepackage{authblk}
\usepackage{amsmath,amsthm}
\usepackage{amsfonts}
\usepackage{graphicx}
\usepackage{enumerate}
\usepackage[usenames,dvipsnames]{color}
\usepackage{subfigure}
\usepackage{color}
\usepackage{apacite}
\usepackage[utf8]{inputenc}
\usepackage[english]{babel}
\usepackage{lineno}

\usepackage{filecontents}

    \def\qed{\hfill$\sqcap\kern-8.0pt\hbox{$\sqcup$}$\\}
        \def\d{{\textnormal d}}

\newtheorem{theorem}{Theorem}
\newtheorem{lemma}{Lemma}
\newtheorem{problem}{Problem}

\newtheorem{corollary}{Corollary}
\theoremstyle{definition}
\newtheorem{definition}{Definition}
\newtheorem{example}{Example}

\begin{document}

\vspace*{0mm}

\noindent
{\Large \bf Tail dependence of the Gaussian copula revisited}

\vspace*{3mm}

\noindent
{\large  Edward Furman$^{a,*}$, Alexey Kuznetsov$^{a}$, Jianxi Su$^{a}$, Ri\v{c}ardas Zitikis$^b$}

\bigskip

\noindent
$^a$ Department of Mathematics and Statistics, York University, Toronto, Ontario M3J 1P3, Canada

\noindent
$^b$ Department of Statistical and Actuarial Sciences,
University of Western Ontario, London, Ontario N6A 5B7, Canada

\vspace*{-2mm}
\noindent
\rule{165mm}{0.2mm}
\\
\noindent
\textbf{Abstract.}
Tail dependence refers to clustering of extreme events. In the context of financial
risk management, the clustering of high-severity risks has a devastating effect on the well-being of firms and is thus of pivotal importance in risk analysis.

When it comes to quantifying the extent of tail dependence, it is generally
agreed that measures of tail dependence must be independent of the marginal distributions
of the risks but rather solely copula-dependent. Indeed, all classical measures of tail dependence
are such, but they investigate the amount of tail dependence
along the main
diagonal of copulas, which has often little in common with the concentration of extremes in the
copulas' domain of definition.

In this paper we urge that the classical measures of tail dependence may
underestimate the level of tail dependence in copulas. {For the Gaussian copula, however, we prove that
the classical measures are maximal. The implication of the result is two-fold: On the one hand, it means
that in the Gaussian case, the (weak) measures of tail dependence that have been reported and used are of
utmost prudence, which must be a reassuring news for practitioners. On the other hand, it further encourages
substitution of the Gaussian copula with other copulas that are more tail dependent.}

\bigskip

{\vskip 0.15cm}


\noindent
\rule{25mm}{0.2mm}
\\
{\footnotesize
$^{*}$Corresponding author. Tel. +1(416)736-2100 ext. 33768. \\
E-mail addresses: efurman@mathstat.yorku.ca (E. Furman), kuznetsov@mathstat.yorku.ca (A. Kuznetsov), gavinsox@mathstat.yorku.ca (J. Su), zitikis@stats.uwo.ca (R. Zitikis).
}

\section{Introduction}
\label{sec-1}
``The devil is in the tails"  is the title of the paper by Donnelly and
Embrechts (2010) who refute the harsh criticism of mathematics (Salmon,
2012) in general, and of the Gaussian copula-based credit risk model
of Li (2000) in particular. As a member of the unholy trinity (Kousky and
Cooke, 2009), the notion of tail dependence is in the very center of this
controversy. Speaking plainly, tail dependence is about the clustering of
extreme events, and it is a rather
daunting phenomenon in the context of risk management. For example, it implies
that devastating losses
within portfolios of risks as well as defaults of financial enterprises in
credit risk portfolios occur together (R\"{u}schendorf, 2013; Wang et al.,
2013; Puccetti and R\"{u}schendorf, 2014).

Mathematically, there exist a variety of ways to quantify the extent of tail dependence
in bivariate random vectors with dependence structures given by copula functions
$C:[0,\ 1]^2\rightarrow [0,\ 1]$ (see, e.g., Nelsen, 2006; Durante and Sempi, 2015, for reviews of the
theory of copulas).
Arguably the most popular
measure of lower tail dependence is nowadays attributed to Joe (1993) (also
Sibuya, 1959) and given by
\begin{eqnarray}
\label{lambda}
\lambda_L:=\lambda_L(C)= \lim_{u \downarrow 0} {C(u,u)\over u}.
\end{eqnarray}
Non-zero (more precisely $(0,\ 1]$) values of index (\ref{lambda})
suggest lower tail dependence in $C$. Just like other synthetic measures, $\lambda_L$ is not always
reliable because it sometimes underestimates the extent of lower tail dependence in copulas
as the next example demonstrates.
\begin{example}
\rm
\label{ExC1}
Consider the following copula (Nelsen, 2006, Example 3.3)
\[
C_{\theta}(u,\ v)=\left\{
\begin{array}{ll}
u, & 0\leq u\leq \theta v\leq \theta, \\
\theta v, & 0\leq \theta v< u < 1- (1-\theta)v, \\
u+v - 1, & \theta \leq 1- (1-\theta) v\leq u\leq 1,
\end{array}
\right.
\]
parametrized by $\theta \in [0,1]$. This copula has two singularities, and it is fully {co-monotonic} (fully {counter-monotonic})
for $\theta=1$ ($\theta=0$, respectively). It is easy to see that
\[
\lambda_L(C_\theta)=\lim_{u\downarrow 0}\frac{\theta u}{u}=\theta.
\]
Let $\lambda_L^\ast(C_{\theta})$ be the measure
as in (\ref{lambda}) but now along the path $\left(u\sqrt{\theta},\ u/\sqrt{\theta}\right)_{0\leq u\leq 1}$ rather than along the diagonal $(u,u)$. Clearly in this case
\[
\lambda_L^\ast(C_\theta)=\lim_{u \downarrow 0}\frac{C(u\sqrt{\theta},u/\sqrt{\theta})}{u}=
\sqrt{\theta}>\theta=\lambda_L(C_\theta),
\]
for every $\theta\in(0,\ 1)$.
\end{example}

On a different note,
when limit (\ref{lambda}) is zero, it is often useful to rely on the somewhat more
delicate index of weak tail dependence $\chi_L\in[-1,\ 1]$ (Coles et al., 1999; Fischer and Klein, 2007)
that is given by
\begin{equation}\label{ind-cl}
\chi_{L}:=\chi_L(C)=\lim_{u\downarrow 0}\frac{2\log u }{\log C(u,u)}-1 ,
\end{equation}
and/or to the index of tail dependence $\kappa_L:=\kappa_L(C)\in[1,\ 2]$ (Ledford and Tawn, 1996) that solves the equation
\begin{equation}\label{ind-kl}
C(u,u)= \ell(u) u^{\kappa_L} \quad \textrm{when} \quad u\downarrow 0,
\end{equation}
assuming that we can find a slowly varying at $0$ function $\ell(u)$.
The following example demonstrates that $\kappa_L$ can also be misleading in a similar way to that of Example \ref{ExC1}.
\begin{example}
\label{ExMO1}
\rm
Recall that the Marshall-Olkin copula is given by
\begin{equation}\label{copula-mo}
C_{a,b}(u,v)=\min(u^{1-a}v,uv^{1-b})
\quad \textrm{for} \quad 0\le u,v\le 1,
\end{equation}
where $a,b\in [0,1]$ are parameters (Cherubini et al., 2013). Denote by $\kappa_L^\ast(C_{a,b})$ a measure
{mimicking} (\ref{ind-kl}) that verifies the tail dependence of $C_{a,b}$ along the path
$\left(u^{2a/(a+b)},u^{2b/(a+b}\right)_{0\leq u\leq 1}$. We readily check that
\[
\kappa_L^\ast(C_{a,b})=2-\frac{2ba}{a+b}\leq 2-\min(a,b)=\kappa_L(C_{a,b}),
\]
where the equality holds only if $a=b$.
\end{example}

Speaking generally, indices (\ref{lambda}), (\ref{ind-cl}) and (\ref{ind-kl}) may underestimate the amount of tail dependence even
in copulas that are symmetric and do not have singularities (Furman et al., 2015).
The reason is that all the aforementioned indices of lower tail dependence rely entirely on the
behavior of copulas along their main diagonal
 $(u,\ u)_{0\leq u\leq 1}$. However,
the tail dependence of copulas can be substantially stronger along the paths other than
 the main diagonal. This can be a serious disadvantage, as reported by, e.g., Schmid and Schmidt (2007), Zhang (2008), Li et al. (2014), and Furman et al. (2015).

When it comes to the bivariate Gaussian copula, henceforth denoted by $C_{\rho}$, which has become a synonym of the recent
subprime mortgage crisis, it can be shown that indices (\ref{lambda}), (\ref{ind-cl}) and
(\ref{ind-kl})
are equal to $0,\ \rho$ and $2/(1+\rho)$, respectively, where $\rho\in(0,\ 1)$ is the
correlation coefficient.
In the light of discussion hitherto, the following most natural problem
arises:
\begin{problem}
\label{pr}
Let $\psi,\varphi:[0,\ 1]\rightarrow [0,\ 1]$ be functions yielding an admissible path
$(\psi(u),\ \varphi(u))_{0\leq u\leq 1}$ in $[0,\ 1]^2$, and let $\lambda_L^\ast$,
$\chi_L^\ast$ and $\kappa_L^\ast$ be the counterparts of (\ref{lambda}), (\ref{ind-cl}) and
(\ref{ind-kl}), respectively, calculated along the noted path. Is there an admissible path such that
any of the following
bounds holds
\begin{equation}
\label{pr-lim}
\lambda_L^\ast(C_\rho)>\lambda_L(C_\rho),\textnormal{ }
\chi_L^\ast(C_\rho)>\chi_L(C_\rho) \textnormal{ and/or }
\kappa_L^\ast(C_\rho)<\kappa_L(C_\rho)?
\end{equation}
\end{problem}
\noindent
A positive answer to this question would reinstate to an extent the Gaussian copula in
public favor, whereas a negative answer would mean that index (\ref{lambda}) is
maximal in the Gaussian case, which of course does not imply the same conclusion for
other copulas.

In this paper we investigate the aforementioned problem. To this end, in Section \ref{sec-2} we set out to formally define the class of `admissible' path functions as well as the collection of
`admissible' paths mentioned in Problem \ref{pr}.
In Section \ref{sec-3}, we then provide a complete solution to Problem \ref{pr}. Our
proofs rely on subtle geometric arguments involving intersections of convex curves with
their rotations. Section \ref{sec-4} concludes the
paper.

\section{Paths and indices of maximal tail dependence}
\label{sec-2}

Our main goal in this section is to describe admissible paths $(\psi(u),\ \varphi(u))_{0\leq u\leq 1}$ formally.
We borrow heavily from Furman et al. (2015).
\begin{definition}
\label{paths}
A function $\varphi: [0,1] \to [0,1]$ is called {\it admissible} if it satisfies the following conditions:
\begin{enumerate}[\rm (1)]
\item [(C1)]
$\varphi(u) \in [u^2,1] $ for every $u\in [0,1] $;
\item [(C2)]
$\varphi(u)$ and $u^2/\varphi(u) $ converge to $0$ when $u\downarrow 0$.
\end{enumerate}
Then the path $(\varphi(u),u^2/\varphi(u))_{0\le u \le 1}$ is admissible whenever the function $\varphi $ is admissible. Also, we denote by $\mathcal{A}$ the set of all admissible
functions $\varphi$.
\end{definition}

A number of observations are instrumental to clarify the definition. First,
condition (C1) makes sure that $\varphi(u)\in[0,\ 1]$ and
$u^2/\varphi(u) \in [0,\ 1]$, whereas  condition (C2) is motivated by the fact that we are
interested in the behavior of the copula $C$
near the lower-left vertex of its domain of definition. Second, the function $\varphi_0(u)=u,\ u\in[0,\ 1]$, is admissible and yields the main
diagonal $(u,\ u)_{0\leq u\leq 1}$. Third, for the independence
copula {$C^{\perp}$}, it holds that $C^{\perp}(\varphi(u),u^2/\varphi(u))=u^2,\  u\in[0,\ 1]^2$,
which is path-independent as expected, thus warranting the choice
$\psi(u)=u^2/\varphi(u),\ u\in [0,\ 1]$.

In order to determine the strongest extreme co-movements of risks
for any copula $C$, we search for functions $\varphi\in\mathcal{A}$ that maximize the probability
\[
\Pi_{\varphi }(u) =C\big (\varphi(u),u^2/\varphi(u)\big ), \;\;\; u \in (0,1),
\]
or, equivalently, the function
\[
d_\varphi(C,C^\perp)(u)=C\big (\varphi(u),u^2/\varphi(u)\big )-
C^\perp (\varphi(u),u^2/\varphi(u)\big ), \;\;\; u \in (0,1),
\]
which is non-negative for positively quadrant dependent (PQD) (Lehmann, 1966)
 copulas $C$.
Then an admissible function $\varphi^* \in \cal{A} $ is called {\it a function of maximal dependence} if
\begin{equation}
\label{Pi-m}
\Pi_{\varphi^* }(u)=\max_{\varphi \in \cal{A} }\Pi_{\varphi }(u)
\end{equation}
for all $u\in (0,1)$. The corresponding admissible path
$(\varphi^*(u),u^2/\varphi^*(u))_{0\leq u\leq 1}$ is called
{\it a path of maximal dependence}.
Generally speaking, the path $\varphi^\ast$ is not unique, but for each such path the value of $\Pi_{\varphi^\ast}$ is the same. In what follows, we use the notation $\Pi^*(u)$ instead of $\Pi_{\varphi^* }(u)$.

Given the new paradigm of prudence that has taken the world of quantitative risk management by storm (OSFI, 2015), it is sensible to introduce conservative variants of indices
(\ref{lambda}), (\ref{ind-cl}) and (\ref{ind-kl}) that would rely on path of maximal dependence (\ref{Pi-m}), rather than on the main
diagonal path of the copula $C$. Namely, we suggest
\begin{equation}
\lambda_L^*:=\lambda_L^*(C)= \lim_{u \downarrow 0} {\Pi^*(u)\over u}
\textnormal{ instead of }
\lambda_{L}(C)=\lim_{u \downarrow 0} \frac{C(u,u)}{u},
\end{equation}
and
\begin{equation*}
\label{chai}
\chi_L^*:=\chi_L^*(C)=\lim_{u\downarrow 0}\frac{2\log u }{\log \Pi^*(u)}-1
\textnormal{ instead of }
\chi_L=\lim_{u\downarrow 0}\frac{2\log u }{\log C(u,\ u)}-1,
\end{equation*}
subject to the existence of the limits, and also
\begin{equation*}
\label{mo-1b}
\Pi^*(u)=\ell^*(u) u^{\kappa_L^*},\ u\downarrow 0
\textnormal{ instead of }
\Pi(u)=\ell(u) u^{\kappa_L},\ \quad u\downarrow 0,
\end{equation*}
assuming that there exist slowly varying at zero functions $\ell^\ast(u)$ and $\ell(u)$
(Ledford and Tawn, 1996).
These new indices of tail dependence provide a more prudent
estimation of the extent of tail dependence in copulas and, in conjunction with tail-based
risk measures, are capable of distinguishing
between risky positions in situations where the classical indices of tail
dependence fail to do so
(Furman et al., 2015, Section 3).

A useful technique for deriving function(s) of maximal dependence, and thus in turn of the corresponding
indices, consists of three
steps:
\begin{enumerate}
  \item [(S1)] searching for critical points of the function $x \mapsto C(x,u^2/x)$ over the interval $[u^2,1]$ and for each $u\in [0,\ 1]$;
  \item [(S2)] checking which solution(s) is/are global maximum/maxima;
  \item [(S3)] verifying that the function $u\mapsto \varphi^*(u)$  is in $\mathcal{A}$.
\end{enumerate}
Accomplishing these tasks sometimes results in explicit formulas for maximal dependence functions, while in
some other cases obtaining closed-form solutions may not be possible. {For example, as we see from Furman et al. (2015), for the Farlie-Gumbel-Morgenstern (FGM) copula this task is doable, whereas for the generalized Clayton copula there is no closed-form solution.
We refer the reader to, respectively, Sections 4 and 6 in Furman et al. (2015) for more technical discussions.}

Sometimes, especially when formulas for conditional copulas
are readily available, it is useful to recall that partial derivatives of copulas are
conditional copulas, and thus the task of determining the set of critical points becomes
equivalent to
finding all the solutions in $x\in[u^2,\ 1]$ to the equation
\begin{equation}
\label{cond_c}
xC_{2|1}\left(\frac{u^2}{x}|x\right)= \frac{u^2}{x}C_{1|2}\left(x|\frac{u^2}{x}\right).
\end{equation}

{Interestingly, as has been also pointed out by one of the referees, there
are symmetric copulas whose paths of maximal dependence are diagonal:
e.g., the FGM and Clayton copulas, as well as the symmetric subclass of
the Marshall-Olkin copulas (see, Furman et al., 2015). Hence, a
natural question is whether or not the path of maximal dependence always
coincides with the diagonal when the copula function is symmetric?
Unfortunately, the answer to this question is not always positive, as we
 illustrate next.

First, we recall from Furman et al. (2015) that there are symmetric copulas whose paths of maximal dependence are not diagonal, such as the $0.5/0.5$ mixture of two `mirrored' Marshall-Olkin copulas.  Next we present an example showing that this argument also holds for the absolutely continuous subclass of symmetric copulas.

For this, we recall the bivariate extreme value copula (Pickands, 1981; Guillotte and Perron, 2016)
\[
C_A(u,v)=\exp\left\{\ln(uv){A\left(\frac{\ln v}{\ln uv} \right)}\right\},
\]
where $A:[0,1]\rightarrow [0.5,1]$ is the Pickands dependence function, which is convex and satisfies the bounds  $(1-t)\vee t \leq A(t) \leq 1$ for $t\in [0,1]$. Define the $0.5/0.5$ mixture of two `mirrored' extreme value copulas by
\begin{eqnarray}
\label{mix-extreme-c}
C_{A_1,A_2}(u,v)=\frac{1}{2}\big(C_{A_1}(u,v)+C_{A_2}(u,v) \big),
\end{eqnarray}
where $A_1$ and $A_2$ are two Pickands dependence functions such that $A_1(t)=A_2(1-t)$ for $t\in [0,1]$.  It is not difficult to see that copula (\ref{mix-extreme-c}) is PQD, symmetric around the diagonal, and absolutely continuous when the Pickands dependence functions $A_1$ and $A_2$ are differentiable.

To prove that the path of maximal dependence for the just defined copula may not be diagonal, it is sufficient to show that there exists $A_1$ such that
\begin{eqnarray}
\label{ineqn-extreme-c}
\frac{\partial^2}{\partial x^2}C_{A_1,A_2}(x,u^2/x)\Big|_{x=u} > 0.
\end{eqnarray}
Equivalently, we need to verify that for $\psi(t,u)=u^{A_1(1-t)}+u^{A_1(t)}$ and $u\in [0,\ 1]$ we have
\begin{eqnarray}
\frac{\partial^2}{\partial t^2} \psi(t,u)\big|_{t=0.5}
=2 u^{A_1(1/2)}\ln (u)\left(
(A_1^{(1)}(1/2))^2\ln u+A_1^{(2)}(1/2)
\right)
>0,
\label{deriv-3}
\end{eqnarray}
where $A_1^{(k)}$ is the $k$-th derivative of $A_1$. Hence, unless the function $A_1$ attains its minimum at $t=0.5$, there exists $u^\ast\in[0,\ 1]$ such that statement (\ref{deriv-3}) holds for all $u\in[0,\ u^\ast]$.
This suggests that $x\mapsto C(x,\ u^2/x)$ is convex at $x=u$ for $u\in[0,\ u^\ast]$, and so the path of maximal dependence
cannot coincide with the diagonal on the aforementioned
interval, that is, $\varphi^\ast(u)\neq u$ for $u\in[0,\ u^\ast]$.

We conclude this section by noting that other scholars have also
considered  other than the diagonal paths when measuring tail dependence. For
example, Asimit et al. (2016) use the conditional Kendall's tau to measure
tail dependence. Joe et al. (2010) introduce the tail dependence function $b(w_1,w_2;C)=\lim_{u\downarrow 0}C(uw_1,uw_2)/u$ for $w_1,w_2>0$ to measure tail dependence via different directions.  Hua and Joe (2014) use the excess-of-loss economic
pricing functional to study tail dependence. All of these measures as well
as the notion of maximal tail dependent discussed in the present paper
provide complementary ways for understanding tail dependence.
}

\section{Main results}
\label{sec-3}
The bivariate Gaussian copula arises from the bivariate normal distribution. As such,
it is arguably the most popular and well-studied copula, extensively used in financial and
insurance mathematics (MacKenzie and Spears, 2014).
We recall that the Gaussian copula $C_{\rho}(u,v)$ is defined, for $0\leq u,v\leq 1$,
 as follows
\begin{equation}\label{def_CG}
C_{\rho}(u,v)=\Phi_2(\Phi^{-1}(u),\Phi^{-1}(v);\rho),
\end{equation}
where $\Phi(u)$ and $\Phi^{-1}(u)$ are the standard-normal distribution function and its inverse, and
\begin{equation}\label{def_Phi2}
\Phi_2(s,t;\rho)=\int_{-\infty}^{s} \int_{-\infty}^{t}  \frac{1}{{2\pi}\sqrt{1-\rho^2}}\exp \left \{-\frac{x^2-2\rho x y +y^2}{2(1-\rho^2)}\right \}dy dx
\end{equation}
is the distribution function of the bivariate normal distribution with correlation parameter  $\rho\in (-1,1)$, defined for all $s,t\in\mathbb{R}$.

\begin{theorem}\label{th}
For the Gaussian copula $C_{\rho}$,
\begin{enumerate}[\rm (I)]
\item \label{th-a}
when $\rho \in (-1,0)$, there is no admissible path of maximal dependence;
\item \label{th-b}
when $\rho=0$, every admissible path is a path of maximal dependence;
\item \label{th-c}
when $\rho \in (0,1)$, the only path of maximal dependence is the diagonal $(u,u)_{0\leq u \leq 1}$.
\end{enumerate}
\end{theorem}
\begin{corollary}\label{cor}
For the Gaussian copula $C_{\rho}$,  when $\rho \in [0,1)$ we have
\begin{enumerate}[\rm (A)]
\item $\lambda_L^*(C_{\rho})=\lambda_L(C_{\rho})=0$;
\item $\chi_L^*(C_{\rho})=\chi_L(C_{\rho})=\rho$;
\item $\kappa_L^*(C_{\rho})=\kappa_L(C_{\rho})=2/(1+\rho)$.
\end{enumerate}
\end{corollary}

\begin{proof}[Proof of Theorem \ref{th}]
The proof of parts (\ref{th-a}) and (\ref{th-b}) of Theorem \ref{th} is simple. When $\rho \in (-1,0)$ we have $C_{\rho}(u,v)< uv$ for all $u,v\in (0,1)$. Therefore, $\Pi_{\varphi }(u)$ achieves its maximum at either $\varphi(u)=u^2$ or $\varphi(u)=1$. However, the two paths $(u^2,1)_{0\leq u \leq 1}$ and $(1,u^2)_{0 \leq u \leq 1}$ are not admissible, which establishes statement (\ref{th-a}). Statement (\ref{th-b}) follows from the fact that in the case $\rho=0$ the Gaussian copula  $C_{\rho} $ reduces to the independence copula  $C^\perp(u,\ v)=uv,\ 0\leq u,v,\leq 1$.

The proof of part (\ref{th-c}) of Theorem \ref{th} is much more involved, and it requires several auxiliary results.
Our first goal is to rephrase the statement about the location of paths of maximal dependence as a geometric statement about certain curves. To this end, for $\alpha \in (0,1)$, we define
\begin{equation}\label{def_C_alpha}
{\cal{C}}_{\alpha}:=\left \{(w,z)  \;:\;  \Phi(w)\Phi(z)=\alpha \right \} \subset {\mathbb R}^2.
\end{equation}
These are the level sets of the function of two variables  $(w,z) \mapsto \Phi(w) \Phi(z)$, and these sets play a pivotal role in our proof.

Given a point $(x_1,x_2) \in {\mathbb R}^2$, we denote by ${}^{\beta}(x_1,x_2)$ the point $(y_1,y_2)$ obtained from $(x_1,x_2)$ by rotation by angle $\beta$ counter-clockwise, that is
\begin{align*}
y_1&=\cos(\beta)x_1 - \sin(\beta) x_2, \\
y_2&=\sin(\beta)x_1 +\cos(\beta) x_2.
\end{align*}
Similarly, for any set $\gamma \subset {\mathbb R}^2$ we denote by ${}^{\beta} \gamma$ the result of rotating the set $\gamma$ counter-clockwise by angle $\beta$.

Consider the following statement:
\begin{align}\label{condition}
&{\textnormal{ \it For any $\alpha \in (0,1)$  and $\beta \in (0,\pi/2)$, the intersection
 ${\cal{C}}_{\alpha}\cap {^{\beta}}{\cal{C}}_{\alpha}$}}\\ \nonumber
 &\qquad \qquad\qquad \;\;{\textnormal{\it consists of a unique point.} }
\end{align}
We next prove that the statement above is in fact stronger than statement (\ref{th-c}) of Theorem
\ref{th}. Thus proving \eqref{condition} automatically completes the proof of Theorem \ref{th}(\ref{th-c}).
\begin{lemma}\label{lemma1}
Statement \eqref{condition} implies part (\ref{th-c}) of Theorem \ref{th}.
\end{lemma}
\begin{proof}
Finding a path of maximal dependence is equivalent to solving the following optimization problem: For every (fixed) $u \in (0,1)$ we want to find the maximum of the function
$[u^2,1]\ni x   \mapsto C_{\rho}(x,u^2/x)$. First of all, note that the restriction
$\rho \in (0,1)$ implies $C_{\rho}(u,v)> uv$ for all $u,v\in (0,1)$. From this result we see that
$C_{\rho}(x,u^2/x)>u^2$
for all $x\in (u^2,1)$, and since $C_{\rho}(x,u^2/x)=u^2$ for $x=u^2$ or $x=1$,
we conclude that the function $x \mapsto C_{\rho}(x,u^2/x)$ achieves the
global maximum for some $\tilde x$ in the open interval $(u^2,1)$.
Since the function $x \mapsto C_{\rho}(x,u^2/x)$ is smooth
in the interval $(u^2,1)$, the global maximum must be one of its critical points, so that
$\frac{\d}{\d x} C_{\rho}(x,u^2/x)=0$ at the maximum point $x=\tilde x$.
To find the critical points, we solve the following equation:
\begin{equation}\label{eqn_hprime_1}
x\Phi\left(\frac{1}{\sqrt{1-\rho^2}}(\Phi^{-1}(u^2/x)-\rho\Phi^{-1}(x) ) \right)
=\frac{u^2}{x}\Phi\left(\frac{1}{\sqrt{1-\rho^2}}
(\Phi^{-1}(x)-\rho\Phi^{-1}(u^2/x)) \right)
\end{equation}
in $x\in(u^2,\ 1)$ for $u\in (0,\ 1)$, which was derived by using equation (3.1) and (3.2) in Meyer (2013) (also Fung and Seneta, 2011; McNeil et al., 2005) as well as \eqref{cond_c}, \eqref{def_CG} and \eqref{def_Phi2} in the current paper.
Note that $x=u$ is clearly a solution of equation (\ref{eqn_hprime_1}). Thus to establish
our main result, it is enough to prove that for every fixed $u \in (0,1)$ there are no other
solutions to \eqref{eqn_hprime_1} except for $x=u$.

Also, we introduce the following change of variables: let $\Lambda$ be the map that
sends any point $(u,x)$ from the domain
\begin{equation*}
D=\{ 0 < u < 1, \; \; u^2 < x < 1\} \subset {\mathbb R}^2
\end{equation*}
into another point $\Lambda (u,x)=(w,z) \in {\mathbb R}^2$ according to the rule:
\begin{equation}\label{def_w_z}
w = \Phi^{-1}(x) \;\;  \textnormal{  and }  \;\; z = \frac{1}{\sqrt{1-\rho^2}}\Phi^{-1}(u^2/x)-\frac{\rho}{\sqrt{1-\rho^2}}\Phi^{-1}(x).
\end{equation}
It is easy to see that $\Lambda$ is a diffeomorphism of $D$ onto ${\mathbb R}^2$, with the inverse map  	
$(u,x)=\Lambda^{-1} (w,z)$ given by
\begin{equation}\label{inverse_Lambda}
u=\left[\Phi(w) \Phi(\rho w + \sqrt{1-\rho^2} z )\right]^{1/2}, \;\;\; x=\Phi(w).
\end{equation}
Using (\ref{def_w_z}) it is easy to check that the diagonal $(u,u)$ of the set $D$ is mapped onto the straight line
$$
l:=\left\{ (z,w) \in {\mathbb R}^2 \; : \;  z = \frac{1-\rho}{\sqrt{1-\rho^2}}  w \right\}.
$$
>From \eqref{inverse_Lambda} we find
$$
u^2/x=\Phi(\rho w + \sqrt{1-\rho^2} z)
$$
and
$$
\frac{1}{\sqrt{1-\rho^2}}
(\Phi^{-1}(x)-\rho\Phi^{-1}(u^2/x))=\sqrt{1-\rho^2} w - \rho z.
$$
Combining the above formulas we conclude that equation \eqref{eqn_hprime_1} is equivalent to
$$
\Phi(w) \Phi(z)=\Phi(\sqrt{1-\rho^2} w - \rho z) \Phi(\rho w + \sqrt{1-\rho^2} z).
$$
Finally, we denote $\beta=\arcsin(\rho)$, so that $\rho=\sin(\beta)$ and $\sqrt{1-\rho^2}=\cos(\beta)$,
and rewrite the above equation as
\begin{equation}\label{eqn_phiz_phiw}
\Phi(w)\Phi(z)=\Phi(w') \Phi(z'),
\end{equation}
where $w'=\cos(\beta)w-\sin(\beta) z$ and $z'=\sin(\beta) w + \cos(\beta) z$. Note that
$\beta \in (0,\pi/2)$ and
 the point $(w',z')$ is obtained from the point $(w,z)$ by rotation by the angle $\beta$ counter-clockwise. Using our previous notation we can write $(w',z')={}^{\beta} (w,z)$.

To summarize, we have shown that after a change of variables $\Lambda: D \mapsto {\mathbb R}^2$ equation \eqref{eqn_hprime_1} is equivalently transformed into equation \eqref{eqn_phiz_phiw}. The latter equation has a simple geometric interpretation: a point $(w,z)$ satisfies \eqref{eqn_phiz_phiw} if and only if $(w,z) \in {\cal C}_{\alpha} \cap {}^{\beta}{\cal C}_{\alpha}$, where $\alpha=\Phi(z)\Phi(w)$ and
${\cal C}_{\alpha}$ is the level set defined in \eqref{def_C_alpha}. Next, all points on the straight line $l$ defined above are the solutions
of \eqref{eqn_phiz_phiw}; this is easy to check directly, and it also follows at once from the fact that the points on $l$ are the images $(w,z)=\Lambda(u,u)$ of points on the diagonal
of $D$, which do satisfy the equivalent equation \eqref{eqn_hprime_1}. Therefore, we have shown that for every
$\alpha  \in (0,1)$ there is {\it at least one} point in the intersection ${\cal C}_{\alpha} \cap {}^{\beta}{\cal C}_{\alpha}$. If we assume
that for every $\alpha \in (0,1)$ there exists a {\it unique} point in the intersection ${\cal C}_{\alpha} \cap {}^{\beta}{\cal C}_{\alpha}$,
this would imply that there are no other solutions to \eqref{eqn_phiz_phiw} except for those on the line $l$,
which in turn implies that the diagonal points $(u,u) \in D$ are the only solutions to \eqref{eqn_hprime_1}. Recall that the solutions to
equation \eqref{eqn_hprime_1} give us the critical points of the function $x \mapsto C_{\rho}(x,u^2/x)$. Therefore, if there are no other critical points except for the diagonal ones, then the path of maximal dependence must be diagonal.
\end{proof}

As we demonstrate next, each level set ${\cal C}_{\alpha},\ \alpha\in(0,\ 1)$ is in fact a smooth curve in ${\mathbb R}^2$ and a boundary of a convex set. Some of these curves are shown in Figure \ref{fig1}.
\begin{figure}
\centering
\includegraphics[height =7cm]{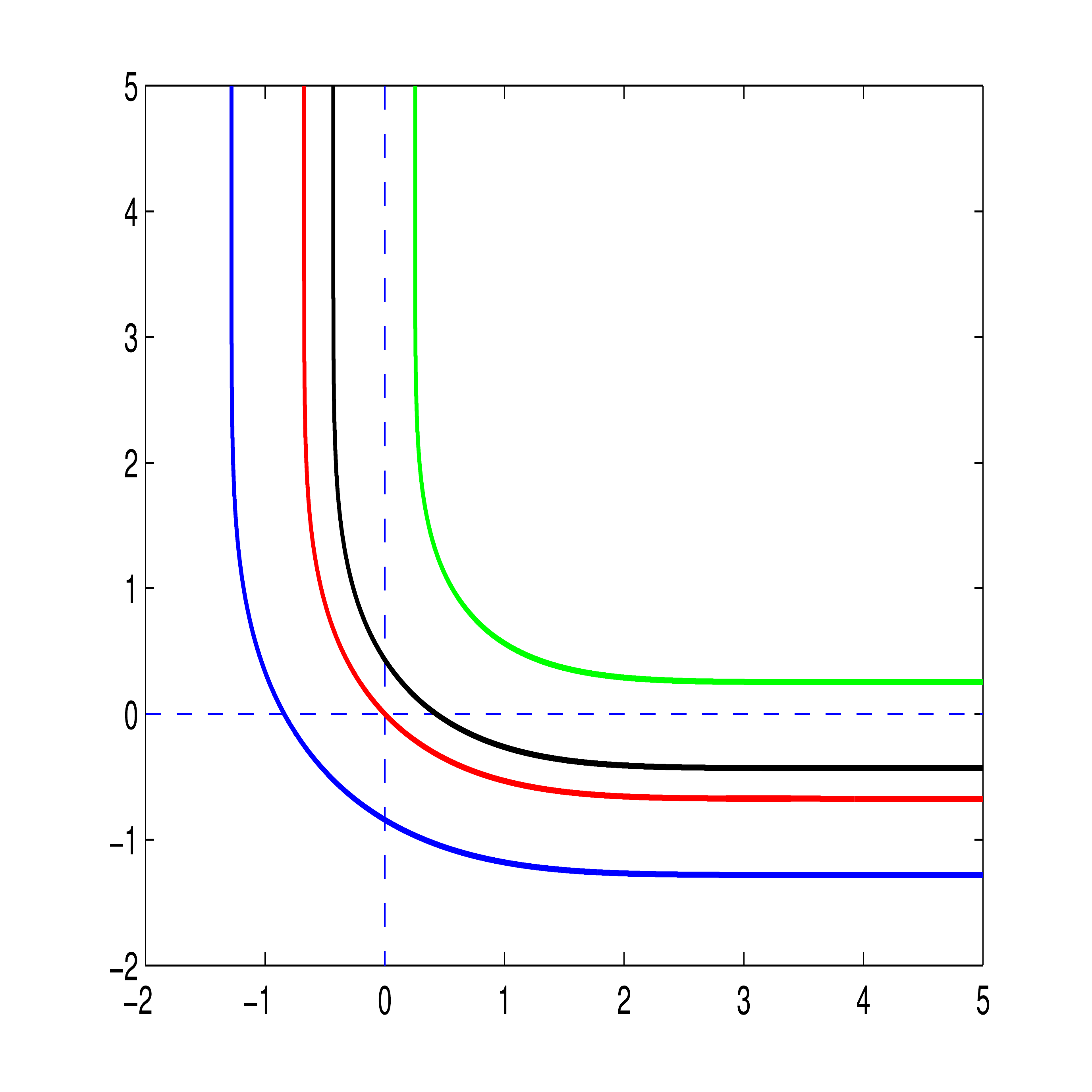}
\caption{The curve ${\cal{C}}_{\alpha}$ for $\alpha = 1/10$  (blue), 1/4 (red), 1/3 (black), and 3/5 (green).}
\label{fig1}
\end{figure}

\begin{lemma}\label{lemma_C_properties} For every $\alpha \in (0,\ 1)$, the following properties hold:
\begin{enumerate}
\item [\textnormal{(P1)}] \label{prop-a}
${\cal{C}}_{\alpha}$ is a smooth curve in ${\mathbb R}^2$;
\item [\textnormal{(P2)}] \label{prop-c}
${\cal{C}}_{\alpha}$ is symmetric with respect to the diagonal line $\{(w,z) : w=z\}\subset {\mathbb R}^2$;
\item [\textnormal{(P3)}] \label{prop-b}
${\cal{C}}_{\alpha}$ is the boundary of a convex set which lies in
$\{w > \Phi^{-1}(\alpha)\} \cap \{z > \Phi^{-1}(\alpha)\}$. The lines
$z=\Phi^{-1}(\alpha)$ and $w=\Phi^{-1}(\alpha)$ are the asymptotes of ${\cal C}_{\alpha}$.
\end{enumerate}
\end{lemma}
\begin{proof}
First of all, we note that the equation $\Phi(w)\Phi(z)=\alpha$ can be solved for $z$ in terms of $w$ as follows
\begin{equation}\label{fn_z_w}
z(w)=\Phi^{-1}(\alpha/\Phi(w)).
\end{equation}
The level set ${\cal C}_{\alpha}$ is simply the graph of $z(w)$, which is clearly a
smooth function defined for $w>\Phi^{-1}(\alpha)$.  This proves property (P1).  The symmetry with respect to interchanging $w \leftrightarrow z$ follows at once from
definition  \eqref{def_C_alpha}.  This proves property (P2).

Since
$$
z(w) \to \Phi^{-1} (\alpha),\ {\text{for}}\ w\to +\infty,
$$
the horizontal line $z=\Phi^{-1}(\alpha)$ is an asymptote, and the vertical asymptote $w=\Phi^{-1}(\alpha)$ follows from the above-mentioned symmetry with respect to $z\leftrightarrow w$.

The function $z(w)$ given by \eqref{fn_z_w} is convex. An easy way to prove this is via the fact that
$\Phi(w)$ is log-concave, which implies that the function of two variables $(w,z) \mapsto \ln(\Phi(z)\Phi(w))$ is concave. Therefore, its upper set
\begin{equation}\label{def_U_alpha}
{\cal U}_{\alpha}:=\{ (w,z) \in {\mathbb R}^2 \; : \; \ln(\Phi(z)\Phi(w))>\ln(\alpha)\}
\end{equation}
must be convex. It is clear that ${\cal C}_{\alpha}$ is the boundary of ${\cal U}_{\alpha}$.  This proves property (P3) and completes the proof of Lemma \ref{lemma_C_properties}.
\end{proof}

We recall that for any set $\gamma  \subset {\mathbb R}^2$ we denote by ${}^{\beta} \gamma$ the rotation
of $\gamma$ by angle $\beta$ counter-clockwise. In particular, if $\gamma$ is a curve which can be written in  polar coordinates $(r, \theta )$ as
\begin{equation}\label{curve-0}
\gamma=\left \{ (r(\theta)\cos(\theta), r(\theta) \sin(\theta)) \; : \;  \theta \in (\theta_1, \theta_2) \right \} ,
\end{equation}
then ${^{\beta}}\gamma$ is also a curve whose expression in polar coordinates is given by
\begin{equation}\label{curve-1}
{^{\beta}}\gamma=\left \{(r(\theta-\beta)\cos(\theta), r(\theta-\beta) \sin(\theta)) \; : \; \theta \in (\theta_1+\beta, \theta_2+\beta) \right \}.
\end{equation}

\begin{figure}
\centering
\includegraphics[height =7cm]{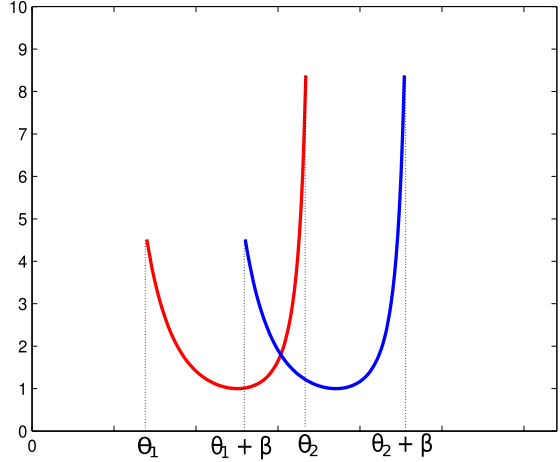}
\caption{Illustration to the proof of Lemma \ref{lemma_inter-point}(i):
 the curves $\gamma$ (red) and ${}^{\beta}\gamma$ (blue) in polar coordinates ($\theta$ is on x-axis, and $r$ is on  y-axis).}
\label{fig2}
\end{figure}

\begin{lemma}\label{lemma_inter-point}
${}$
\begin{itemize}
\item[\textnormal{(i)}]
Assume that a curve $\gamma $ is written in polar coordinates in form (\ref{curve-0}).  If the function $r(\theta)$ is strictly decreasing on $(\theta_1,\tilde \theta)$ and strictly increasing on $(\tilde \theta, \theta_2)$ for some $\tilde \theta \in (\theta_1, \theta_2)$ then for every  $\beta \in (0,2\pi-\theta_2+\theta_1)$ there is at most one point of intersection of $\gamma$ and ${}^{\beta}\gamma$.
\item[\textnormal{(ii)}] Assume that a curve $\gamma$ is given in the parametric form $(w(t),z(t))$, $t\in I$, where $I\subset {\mathbb R}$ is an interval. If
$r(t)=\sqrt{w(t)^2+z(t)^2}$ is nonzero and strictly monotone for $t\in I$, then for any $\beta \in (0,2\pi)$ the curves
$\gamma$ and ${}^{\beta} \gamma$ do not intersect.
\end{itemize}
\end{lemma}
\begin{proof}
While  the proof of part (i) is quite obvious from Figure \ref{fig2}, we present the details of the proof for mathematical rigour. When $\theta_1+\beta>\theta_2$, then it is clear that ${^{\beta}}\gamma \cap \gamma=\emptyset$ because these curves lie in non-intersecting sectors.
We are left with the case when $\theta_1+\beta<\theta_2$. We further restrict ourselves to the case $\theta_1+\beta<\tilde{\theta}$; the argument in the case $\theta_1+\beta \in [\tilde \theta, \theta_2)$ is identical. Denote $r_1(\theta)=r(\theta)$ and $r_2(\theta)=r(\theta-\beta)$.
Since $r_1(\theta)$ is decreasing on $(\theta_1,\tilde \theta)$ we have
$r_2(\theta)>r_1(\theta)$ for $\theta \in (\theta_1+\beta,\tilde \theta)$. Therefore, the curves $\gamma$ and ${^{\beta}}\gamma$
do not intersect when $\theta \in (\theta_1+\beta,\tilde \theta)$.
When $\theta \in (\tilde \theta, \min(\theta_2, \tilde\theta+\beta))$, the function $r_1(\theta)$ is strictly increasing
while $r_2(\theta)$ is strictly decreasing. By considering the values of these two functions at the endpoints of the interval we conclude that there exists a unique number $\theta^*$ for which $r_1(\theta^*)=r_2(\theta^*)$. The point with polar coordinates
$(r_1(\theta^*),\theta^*)$ then gives us the unique point of intersection of $\gamma$ and ${^{\beta}}\gamma$. On the interval $\theta \in [\tilde \theta+\beta,\theta_2)$ we have the inequality $r_1(\theta)>r_2(\theta)$, and so the curves $\gamma$ and ${^{\beta}}\gamma$ do not intersect in this sector. Hence, we have shown that there exists a unique point of intersection of  $\gamma$ and ${^{\beta}}\gamma$
when $\theta_1+\beta<\theta_2$.

To establish part (ii), let us assume that $(\tilde w, \tilde z) \in \gamma \cap {}^{\beta} \gamma$. This condition implies
that a circle $B_R$ with radius $R=\sqrt{\tilde w^2+ \tilde z^2}$ and center at the origin must intersect the curve $\gamma$ at two distinct points. However, this contradicts the condition that
the radius $r(t)$ is strictly monotone along the curve $\gamma$. Thus we have arrived at a contradiction. Therefore, the intersection of $\gamma$ and ${}^{\beta} \gamma$ must be empty.
\end{proof}

\begin{lemma}\label{lemma3}
 Let $z=f(w)$ be a function defined on $0<w<w_0$ (where $w_0$ can be $+\infty$). Assume that $f(w)$ is smooth, decreasing, convex and its graph $\gamma=\{(w,f(w)) \; : \; 0<w<w_0\}$ is symmetric with respect to the line $z=w$. Then the curve $\gamma$ can be written in polar coordinates in form \eqref{curve-0}, and $r(\theta)$ is strictly decreasing on the interval $(0,\pi/4)$ and strictly increasing
on the interval $(\pi/4,\pi/2)$.
\end{lemma}
\begin{proof}
Note that the condition that $\gamma$ is symmetric with respect to the line $z=w$ implies $f(w_0)=0$. This result and the fact that
$f(w)/w$ is strictly decreasing allows us to represent $\gamma$ in  polar coordinates in form \eqref{curve-0} with $\theta_1=0$ and
$\theta_2=\pi/2$. Let $(w_0,z_0)$ be the point of intersection of $\gamma$ and the line $z=w$, so that $z_0=f(w_0)$. The part of the graph
with $w>w_0$ corresponds to the polar coordinate representation with $\theta \in (0,\pi/4)$.
The symmetry of $\gamma$ with respect to $z=w$ implies $f'(w_0)=-1$, and since the function $f(w)$ is convex we see that
$f'(w)>-1$ for $w>w_0$. The radius in polar coordinates is given by $r=\sqrt{w^2+f(w)^2}$. Thus for $w>w_0$,
$$
\frac{\d r}{\d w}=r^{-1} \left(w+f(w)f'(w)\right)>r^{-1} (w-f(w))>r^{-1} (w-f(w_0))=r^{-1}(w-w_0)>0,
$$
where we have used the fact that $f(w)$ is strictly decreasing. This result combined with the fact that
 $\d w/\d \theta<0$ for $\theta \in (0,\pi/2)$
shows that $\d r/\d \theta<0$ for $\theta \in (0,\pi/4)$. The fact that $\d r/\d \theta>0$ for $\theta \in (\pi/4,\pi/2)$ follows
by symmetry.
\end{proof}

We are now ready to complete the proof of part (\ref{th-c}) of Theorem \ref{th}. According to
Lemma \ref{lemma1}, it is enough to establish the validity of statement \eqref{condition}. We do this in four steps, depending on the value of $\alpha\in(0,\ 1)$.

\vspace{0.35cm}
\noindent
{\bf Proof of statement \eqref{condition} for $\alpha \in (0,1/4)$:}

\begin{figure}
\centering
\includegraphics[height =7cm]{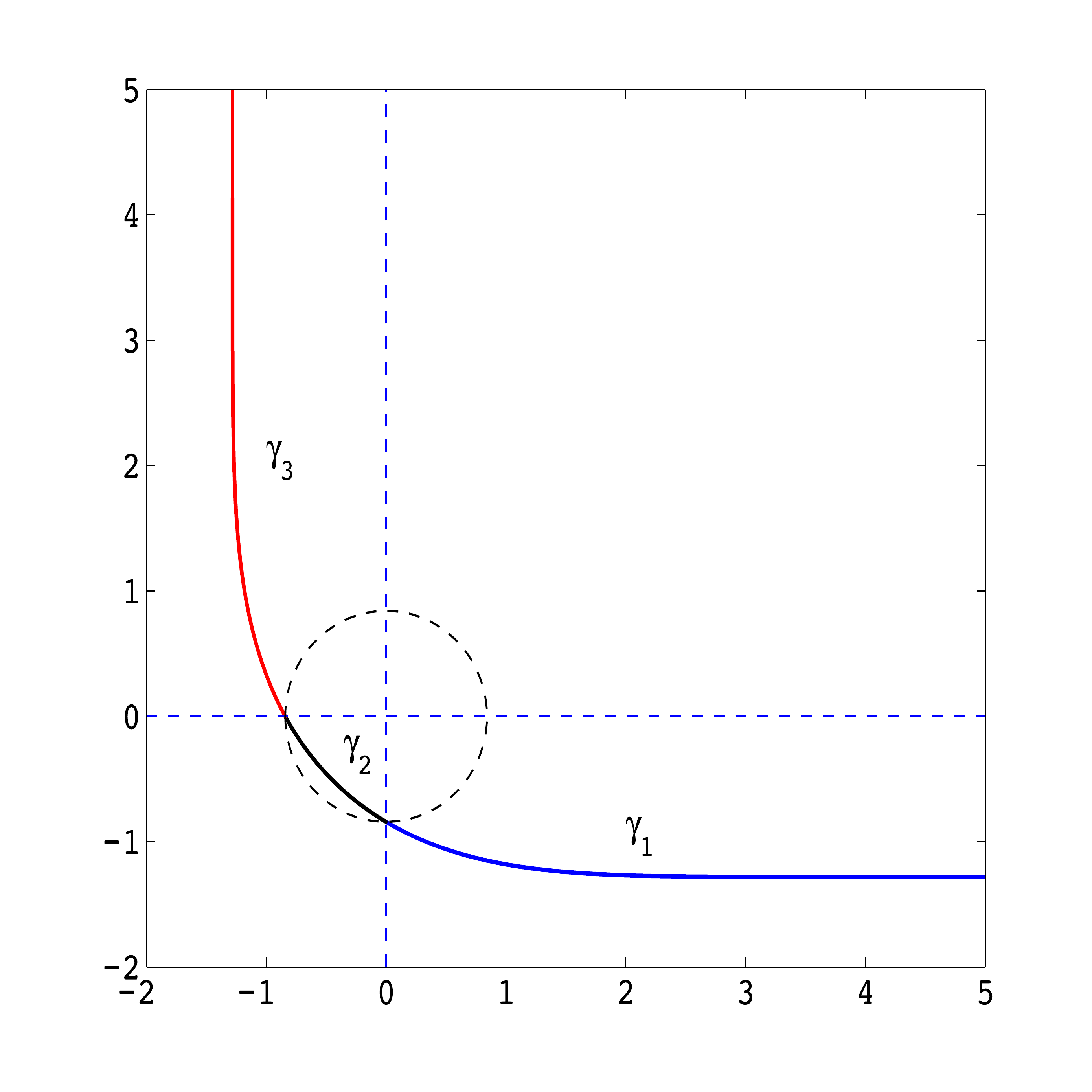}
\caption{Illustration to the proof of statement \eqref{condition} for $\alpha \in (0,1/4)$.}
\label{pic1}
\end{figure}

\noindent
We begin by noting that if $\alpha \in (0,1/4)$ then the origin belongs to the convex set ${\cal U}_{\alpha}$ defined above in
\eqref{def_U_alpha}.   Given this fact and the properties of ${\cal C}_{\alpha}$ which were described in Lemma \ref{lemma_C_properties}, it is clear that the curve ${\cal C}_{\alpha}$ must lie in the second, third and fourth quadrants; see Figure \ref{pic1} or the blue curve in Figure \ref{fig1}). We can express this curve in polar coordinates as follows
\[
{\cal{C}}_{\alpha}=\left \{(r(\theta)\cos(\theta),r(\theta) \sin(\theta)) \; : \;   \pi/2 < \theta <2\pi \right \},
\]
and then divide it into three parts $\gamma_1$, $\gamma_2$ and $\gamma_3$, which lie respectively in the fourth, third and second quadrants (see
Figure  \ref{pic1}).

First we consider the curve $\gamma_1$. We claim that the radius (the distance from the origin) strictly increases as we move along this curve to the right. To see this, we parametrize points on $\gamma_1$ by $(w,z(w))$ for $w>0$, where $z(w)$ is given by  \eqref{fn_z_w}. For $w>0$ the function $z(w)$ is negative and strictly decreasing, which shows that $r(w)=\sqrt{w^2+z(w)^2}$ is strictly increasing.

Next, we consider the curve $\gamma_2$ parametrized by polar coordinates $(r,\theta)$, $\theta \in (\pi,3\pi/2)$. Our goal is to prove that
$r(\theta)$ is strictly decreasing on the interval $(\pi,5\pi/4)$ and strictly increasing on the interval $(5\pi/4,3\pi/2)$.
Note that the function $r(\theta)$ satisfies the equation
$$
\Phi(r(\theta) \cos(\theta))\Phi(r(\theta) \sin(\theta))=\alpha, \;\;\; \pi/2 < \theta <2\pi.
$$
Differentiating both sides of this equation with respect to $\theta$ we obtain
\begin{equation}\label{r_prime_over_r}
\frac{1}{r}\frac{\d r}{\d \theta}= \frac{\Phi(w)w e^{-z^2/2}-\Phi(z)ze^{-w^2/2}}{\Phi(w)z e^{-z^2/2}+\Phi(z)we^{-w^2/2}},
\end{equation}
where we denoted $w=r \cos(\theta)$ and $z=r\sin(\theta)$.
The denominator on the right-hand side of \eqref{r_prime_over_r} is strictly negative in the third quadrant. Thus to prove our claim about the increase/decrease of $r(\theta)$ it is enough to demonstrate that the numerator on the right-hand side of \eqref{r_prime_over_r} satisfies
\begin{align}\label{numerator}
&\Phi(w)w e^{-z^2/2}-\Phi(z)ze^{-w^2/2} >0 \quad \textrm{if} \quad w>z ,\\ \nonumber
&\Phi(w)w e^{-z^2/2}-\Phi(z)ze^{-w^2/2} <0 \quad \textrm{if} \quad w<z .
\end{align}
This is indeed true because the function $h(w)=w e^{w^2/2} \Phi(w)$ is strictly increasing
for all $w\in\mathbb{R}$. The monotonicity of $h(w)$ is obvious for $w>0$ and follows
from Pinelis (2002), in which the monotonicity of
$-h(-w)=w e^{w^2/2} {(1-{\Phi}(w))}$ was studied.


Let us summarize what we have established so far about the curve ${\cal C}_{\alpha}=\gamma_1 \cup \gamma_2 \cup \gamma_3$. As we move along this curve, starting in its upper part, the radius strictly decreases until it reaches its global minimum at the point of intersection of
${\cal C}_{\alpha}$ and the line $z=w$,  and afterwards the radius strictly increases.

Now we are ready to prove that there exists a unique point of intersection between ${\cal C}_{\alpha}$ and ${}^{\beta}{\cal C}_{\alpha}$.
Let us denote $R=-\Phi^{-1}(2\alpha)$, so that $(0,-R)$ is the point of intersection of ${\cal C}_{\alpha}$ and the $z$-axis; this follows from \eqref{fn_z_w}. The monotonicity properties of the radius imply that the curves $\gamma_1$ and $\gamma_3$ lie outside of the circle $B_R$ with the center at the origin and the radius $R$, while the curve $\gamma_2$ lies completely inside this circle (of course the boundaries of these curves meet at the circle). Let us see what happens when we rotate the curve
${\cal C}_{\alpha}=\gamma_1 \cup \gamma_2 \cup \gamma_3$ by angle $\beta \in (0, \pi/2)$ counter-clockwise. The intersection of $\gamma_1 \cap {}^{\beta} \gamma_1$ and $\gamma_3 \cap {}^{\beta} \gamma_3$ is empty due to
Lemma \ref{lemma_inter-point}(ii). The curves ${}^{\beta}\gamma_1$ and $\gamma_3$ (and, similarly, $\gamma_1$ and ${}^{\beta} \gamma_3$) do not intersect since $\gamma_1$ and $\gamma_3$ lie in the fourth and  second quadrants and the angle of rotation $\beta$ is strictly less
than $\pi/2$.  The intersection  $\gamma_1 \cap {}^{\beta} \gamma_2$ and ${\gamma_3} \cap {}^{\beta} \gamma_2$ is empty since these curves lie in different regions separated by the circle $B_R$ (one is inside and the other one is outside of this circle). And finally,
the intersection $\gamma_2 \cap {}^{\beta} \gamma_2$ consists of at most one point due to Lemma \ref{lemma_inter-point}(i). In fact, such a point of intersection must exist since we know that the curves ${\cal C}_{\alpha}$ and ${}^{\beta} {\cal C}_{\alpha}$ do intersect; see the last paragraph of the proof of Lemma \ref{lemma1}. Thus we have proved that for any $\alpha \in (0,1/4)$ and any $\beta \in (0,\pi/2)$ there exists a
unique point of intersection ${\cal C}_{\alpha} \cap {}^{\beta} {\cal C}_{\alpha}$.
\qed

\vspace{0.25cm}
\noindent
{\bf Proof of statement \eqref{condition} for $\alpha=1/4$:}

\noindent
The curve ${\cal{C}}_{1/4}$ contains the origin; see the red curve in Figure \ref{fig1}.
 The proof of statement \eqref{condition} is the same as the proof  above  in the case $\alpha \in (0,1/4)$, except that now the curve $\gamma_2$ degenerates to a single point $(0,0)$  so that  ${\cal C}_{1/4} \cap {}^{\beta} {\cal C}_{1/4}=\{(0,0)\}$.
\qed

\vspace{0.25cm}
\noindent
{\bf Proof of statement \eqref{condition} for $\alpha \in (1/4,1/2)$:}

\begin{figure}
\centering
\includegraphics[height =7cm]{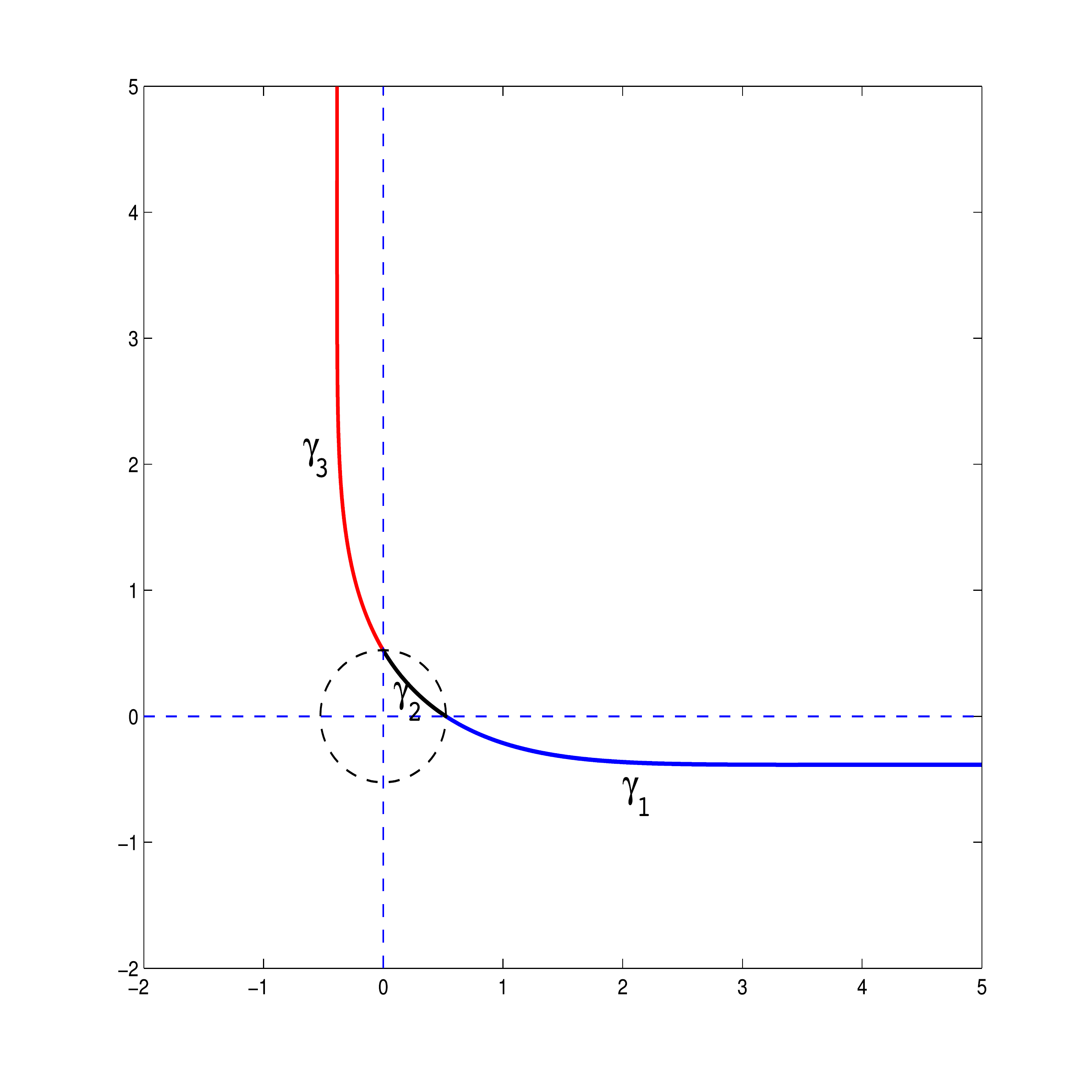}
\caption{Illustration to the proof of statement \eqref{condition} for $\alpha \in (1/4,1/2)$.}
\label{pic2}
\end{figure}

\noindent
When $\alpha \in (1/4,1/2)$ the curve ${\cal C}_{\alpha}$ lies in the first, second and fourth quadrants; see
Figure \ref{pic2} or the black curve
in Figure \ref{fig1}. The proof
of statement \eqref{condition} is the same as in the case $\alpha \in (0,1/4)$, except that we now use Lemma \ref{lemma3} to prove that
the radius $r(\theta)$ is strictly decreasing for $\theta \in (0,\pi/4)$ and strictly increasing for
$\theta \in (\pi/4,\pi/2)$.
\qed

\vspace{0.25cm}
\noindent
{\bf Proof of statement \eqref{condition} for $\alpha \in [1/2,1)$:}

\noindent
When $\alpha \in [1/2,1)$ the curve ${\cal C}_{\alpha}$ lies entirely in the first quadrant;
see the green curve in Figure \ref{fig1}.
In this case the proof of statement \eqref{condition} follows from Lemma \ref{lemma_inter-point}(i) and Lemma \ref{lemma3}.
\qed

This ends the proof of statement (\ref{condition}), and according to Lemma \ref{lemma1}, the proof of Theorem \ref{th} is now complete.
\end{proof}

\section{Concluding comments}
\label{sec-4}

Traditional methods in insurance and finance work well for symmetric and light-tailed risks. It is however
a well-known empirical fact that in reality risks are skewed (Graham and Harvey, 2001), and it is exactly the severe
tail risks that drive economic capital allocations in portfolios of risks. It is not surprising
therefore that the notions of asymmetry and `fat tails'  have been
gaining unprecedented popularity among theoreticians and practitioners (e.g., Staudt, 2010).

The phenomenon of dependent tail risks is equally subtle. In this respect, the not-too-distant
financial crisis doubtlessly demonstrated that, e.g., dependent defaults may be disastrous for economies of entire countries. However, the quantification of tail dependence is not a simple problem. In fact, the
classical approaches that are nowadays commonly employed seem to often underestimate the amount of tail dependence, as they rely solely
 on the main
 diagonal of the copula whereas the copula's (tail) behavior can be very different otherwise.
 For this reason, the aforementioned approaches can miss the so-called
  maximal tail dependence
even in some symmetric dependence structures.

Many, if not the majority, of the models in financial theory have been built with the Gaussian
distribution in mind, for which in this paper we have established that all of the classical
indices of tail dependence (Joe, 1993; Ledford and Tawn, 1996; Coles et al., 1999; Fischer and Klein, 2007) are maximal and
thus conform to the prudence-oriented character of current regulations (e.g., OSFI, 2015).
As the Gaussian copula has been very popular, and it will likely
remain such in the foreseeable future (e.g., MacKenzie and Spears, 2014), our findings are reassuring news for practitioners.

\section*{Acknowledgments}
{
We thank the anonymous referees for valuable comments
and suggestions that improved the work  and resulted in a better
presentation of the material.
We are grateful to Prof. Dr. Paul Embrechts and all participants
of the ETH Series of Talks in Financial and Insurance Mathematics for feedback and
insights.

Our research has been supported by the Natural Sciences and Engineering Research Council (NSERC) of Canada. Jianxi Su also acknowledges the financial support of the Government of Ontario and MITACS Canada via, respectively, the Ontario Graduate Scholarship program
and the Elevate Postdoctoral fellowship.}

\section*{References}
\def\hang{\hangindent=\parindent\noindent}
\hang
{Asimit, V., Gerrard, R., Yanxi, H., Peng, L., 2016.  Tail dependence measure for examining financial extreme co-movements.  Journal of Econometrics, forthcoming.}

\hang
Cherubini, U., Durante, F., Mulinacci, S. (Eds.), 2013. Marshall--Olkin Distributions - Advances in Theory and Applications. Springer, {Switzerland}.

\hang
Coles, S., Heffernan, J., Tawn, J., 1999. Dependence measures for extreme value analyses. Extremes 2 (4), 339--365.

\hang
Donnelly, C., Embrechts, P., 2010. The devil is in the tails: Actuarial mathematics and the subprime mortgage crisis. ASTIN Bulletin 40 (1), 1--33.

\hang
Durante, F., Sempi, C., 2015. Principles of Copula Theory. Chapman and Hall/CRC, London.

\hang
Fischer, M.J., Klein, I., 2007. Some results on weak and strong tail dependence coefficients for means of copulas.
Diskussionspapiere No. 78/2007 // Friedrich-Alexander-Universit\"{a}t Erlangen-N\"{u}rnberg, Lehrstuhl f\"{u}r
Statistik und \"{O}konometrie, {available at http://econstor.eu/bitstream/10419/29623/1/614058171.pdf, accessed on March 3, 2016}.

\hang
Fung, T. and Seneta, E., 2011. The bivariate normal copula is regularly varying. Statistics and Probability Letters 81 (11), 1670--1676.

\hang
Furman, E., Su, J., Zitikis, R., 2015. Paths and indices of maximal tail dependence. ASTIN Bulletin, forthcoming.

\hang
Graham, J.R., Harvey, C.R., 2001. Expectations of equity risk premia, volatility and asymmetry from a corporate finance perspective. Working Paper No. 8678, National Bureau of Economic Research.

\hang
{Guillotte, S., Perron, F., 2016. Polynomial Pickands functions. Bernoulli 22 (1), 213-241.}

\hang
{ Hua, L., Joe, H., 2014.  Strength of tail dependence based on conditional tail expectation.  Journal of Multivariate Analysis 123, 143--159.}

\hang
Joe, H., 1993. Parametric families of multivariate distributions with given margins. Journal of Multivariate Analysis 46 (2), 262--282.

\hang
{Joe, H., Li, H., Nikoloulopoulos, A.K., 2010. Tail dependence functions and vine copulas. Journal of Multivariate Analysis 101, 252--270.}

\hang
Kousky, C., Cooke, R.M., 2009. The unholy trinity: Fat tails, tail dependence, and micro-correlations. Discussion Paper, Resources for the Future, Washington DC.

\hang
Ledford, A.W., Tawn, J.A., 1996. Statistics for near independence in multivariate extreme values. Biometrika 83 (1), 169--187.

\hang
Lehmann, E., 1966. Some concepts of dependence. Annals of Mathematical Statistics 37 (5), 1137--1153.

\hang
Li, D.X., 2000. On default correlation: a copula function approach. Journal of Fixed Income 9 (4), 43--54.

\hang
Li, L., Yuen, K.C., Yang, J., 2014. Distorted mix method for constructing copulas with tail dependence. Insurance: Mathematics and Economics 57, 77--89.

\hang
MacKenzie, D., Spears, T., 2014. `A device for being able to book P\&L': The organizational embedding of the Gaussian copula. Social Studies of Science 44 (3), 418--440.

\hang
McNeil, A.J., Frey, R., Embrechts, P., 2005. Quantitative Risk Management. Princeton University Press, Princeton.

\hang
Meyer, C., 2013. The bivariate normal copula. Communications in Statistics - Theory and Methods 42 (13), 2402--2422.

\hang
Nelsen, R.B., 2006. An Introduction to Copulas, second edition. Springer, New York.

\hang
OSFI, 2015. Own risk and solvency assessment (ORSA). Office of the Superintendent of Financial Institutions, Government of Canada, Ottawa, available at http://www.osfi-bsif.gc.ca/eng/fi-if/rg-ro/gdn-ort/gl-ld/Pages/e19.aspx, accessed on June 17, 2015.

\hang
{Pickands, J., 1981. Multivariate extreme value distributions. Bulletin of the International Statistical Institute 49, 859--878.}

\hang
Pinelis, I., 2002. Monotonicity properties of the relative error of a pad\'{e} approximation for {Mills'} ratio. Journal of Inequalities in Pure and Applied Mathematics 3 (2), 1--8.

\hang
Puccetti, G., R\"{u}schendorf, L., 2014. Asymptotic equivalence of conservative value-at-risk and expected shortfall-based capital charges. Journal of Risk 16 (3), 3--22.

\hang
R\"{u}schendorf, L., 2013. Mathematical Risk Analysis. Springer, Berlin.

\hang
Salmon, F., 2012. The formula that killed Wall Street. Significance 9 (1), 16--20.

\hang
Schmid, F., Schmidt, R., 2007. Multivariate conditional versions of Spearman's rho and related measures of tail dependence. Journal of Multivariate Analysis 98 (6), 1123--1140.

\hang
Sibuya, M., 1959. Bivariate extreme statistics. Annals of the Institute of Statistical {Mathematics} 11 (2), 195--210.

\hang
Staudt, A., 2010. Tail risk, systemic risk and copulas. Casualty Actuarial Society E-Forum 2, 1--23.

\hang
Wang, R., Peng, L., Yang, J., 2013. Bounds for the sum of dependent risks and worst value-at-risk with monotone marginal densities. Finance and Stochastics 17 (2), 395--417.

\hang
Zhang, M-H., 2008. Modelling total tail dependence along diagonals. Insurance: Mathematics and Economics 42 (1), 73--80.

\end{document}